\documentclass{article}
\usepackage[utf8]{inputenc}
\usepackage{a4wide}
\usepackage{xcolor}
\usepackage{amsmath}
\usepackage{lscape}
\usepackage[font=small,labelfont=bf]{caption}
\usepackage{graphicx}
\usepackage{subcaption}
\usepackage{comment}
\usepackage{tabularx}
\usepackage{multirow}
\usepackage{algpseudocode}
\usepackage{algorithm} 
\usepackage{array}
\usepackage{stfloats}
\usepackage{amsmath,amssymb,amsfonts,amsthm} 
{
\newtheorem{theorem}{Theorem}

\theoremstyle{remark}
\newtheorem {remark}{Remark}

}
\usepackage[utf8]{inputenc}
\usepackage{amssymb}
\usepackage{multicol}
\usepackage{geometry}
\usepackage{hyperref}
\usepackage{authblk}
\usepackage{ulem}
\usepackage{moreverb}
\usepackage{amsmath,bm}
\usepackage{graphicx}
\usepackage{siunitx}
\usepackage{subcaption}

\title{Differential Game Strategies for Social Networks with Self-Interested Individuals} 

\author{Hossein B. Jond}
\affil{Department of Cybernetics, Faculty of Electrical Engineering, Czech Technical University in Prague, Czech Republic}

\begin{document}

\maketitle

\begin{abstract}
A social network population engages in collective actions as a direct result of forming a particular opinion. The strategic interactions among the individuals acting independently and selfishly naturally portray a noncooperative game. Nash equilibrium allows for self-enforcing strategic interactions between selfish and self-interested individuals. This paper presents a differential game approach to the opinion formation problem in social networks to investigate the evolution of opinions as a result of a Nash equilibrium. The opinion of each individual is described by a differential equation, which is the continuous-time Hegselmann-Krause model for opinion dynamics with a time delay in input. The objective of each individual is to seek optimal strategies for her own opinion evolution by minimizing an individual cost function. Two differential game problems emerge, one for a population that is not stubborn and another for a population that is stubborn. The open-loop Nash equilibrium actions and their associated opinion trajectories are derived for both differential games using Pontryagin's principle. Additionally, the receding horizon control scheme is used to practice feedback strategies where the information flow is restricted by fixed and complete social graphs as well as the second neighborhood concept. The game strategies were executed on the well-known Zachary's Karate Club social network. The resulting opinion trajectories associated with the game strategies showed consensus, polarization, and disagreement in final opinions.
 
\end{abstract}
\maketitle

\section{Introduction}\label{Int}
Opinion formation is one of the salient features of social networks. Users of a network engage through communication mediums and tools and fuse their opinions with those of other users. However, little is known regarding the strategies employed by network users. In recent decades, there has been increasing research attention on opinion dynamics which is the study of the evolution of public opinions~\cite{LIANG2016112,Friedkin}.

In many of the models proposed for modeling opinion dynamics, each social network user holds a numerical opinion and fuses that opinion through repeated averaging with her neighbors. Among those, the bounded confidence models of opinion dynamics have attracted more attention~\cite{DITTMER20014615,Hegselmann-Krause}. The bounded confidence concept was incorporated into opinion dynamics by Hegselmann and Krause~\cite{Hegselmann-Krause}. In the Hegselmann-Krause (HK) model, a finite number of agents fuse their opinions only with those of others whose opinions do not differ more than their confidence bound. The rationale behind the confidence bound concept is that a user adopts the opinions of those she considers trustworthy. Several works have studied the extensions of the HK model under various settings~\cite{Chen2020,Lorenz2008HeterogeneousBO,VASCA2021109683}. 

Game theory is a promising tool to model and analyze the strategic interactions among users with self-interest, conflict, cooperation, and competition traits in social groups from a mathematical point of view~\cite{osborne2009introduction}. In the context of a game, a user's opinion does not just evolve through a repeated averaging process. Instead, a user's strategy for fusing opinions with others is based on her personal optimization, which models the reward or loss for each possible combination of strategic interactions between her and her neighbors. Each user's strategies must be the best outcome for her optimization and meanwhile self-enforcing, taking other users' strategies into account. Nash equilibrium guarantees all users adopt a self-enforcing plan of action to comply with. The social actors complying with the Nash equilibrium move toward their neighbors’ opinions, but in general not sufficiently enough to achieve the globally minimum social cost~\cite{BINDEL2015248}.

In the past two decades, game theory has been used to study opinion dynamics in social networks. Static games~\cite{Bhawalkar2488615,Grabisch0853,GHADERI20143209,FERRAIOLI201696,Etesami} or evolutionary games~\cite{DI-MARE,DING20101745,HUANG2023113215,Wan9915429} were used in early works and most of the current literature. Nevertheless, only a few works have theoretically analyzed opinion dynamics using differential games and mean-field games. These classes of games were proven to provide mathematically rigorous models of opinion formation in social networks~\cite{Bauso140985676,Banez9754253,Wang2020RobustMF}. 

Opinion dynamics in the context of differential or dynamics games were noticed very recently. Opinion dynamics in a social network on graphs were treated with a dynamic game in~\cite{Jiang2021}. Feedback Nash equilibrium strategies for complete and incomplete communication graphs were found analytically and numerically, respectively. Competitive influence maximization in opinion dynamics using dynamic games was studied in~\cite{Etesami2021OpenLoopES}. An efficient algorithm solves the game for its open-loop equilibrium strategies. A differential game model of opinion dynamics with an open-loop information structure was introduced in~\cite{Niazi2021ADG}. The model was extended in~\cite{YildizO21} to investigate the opinion behavior of stubborn agents in a social network in the presence of a troll.

In this paper, I use noncooperative differential games to model and analyze the opinion dynamics of self-interested users in a social network. The opinion fusion rules are according to the HK opinion dynamics model and are restricted by a communication graph. In particular, I propose two differential game problem extensions to the HK opinion dynamics. One is for a social network that is not stubborn, and the other is for a network that is stubborn. In a stubborn social network, each user is trying to minimize her disagreement with her neighbors' opinions and with her own initial opinion. Stubborn users in social networks were the subject of extensive research~\cite{GHADERI20143209,TIAN2018213}. The explicit game strategies of Nash equilibrium and their corresponding opinion trajectories for both differential games were obtained analytically. Three receding horizon implementations of the open-loop game strategies were proposed to practice feedback strategies for the non-stubborn social network. 

The main contributions of this paper in comparison with closely related works in \cite{Niazi2021ADG,YildizO21} are fourfold. First, \cite{Niazi2021ADG,YildizO21} considered the single integrator dynamics model for opinion evolution, whereas I utilize the HK opinion dynamics model. Unlike the single integrator model, the HK model involves the topological structure of the social network. Second, the cost functions in \cite{Niazi2021ADG,YildizO21} involve minimizing disagreements in the network throughout the entire opinion formation process. In this work, disagreements in the final opinion are taken into account in the formulation of cost functions. The selection of cost functions in this manner leads to more compact explicit expressions for the Nash equilibrium and the opinion trajectories associated with it, which in turn facilitates acknowledging the emergence of distributed information fashion in the solution. Third, I consider an input time delay in the formulation of the differential games. The confidence bound concept in the HK opinion dynamics model is a matter of feedback strategies. Finally, I introduced the receding horizon implementation of open-loop game strategies to practice feedback strategies. Even though a model predictive control (MPC) scheme was used in~\cite{Wongkaew} to construct feedback control strategies for the HK opinion formation model in an optimal control framework, the MPC implementation required solving a sequence of open-loop optimality systems discretized by an appropriate Runge–Kutta scheme and solved by a nonlinear conjugate gradient method. The proposed receding horizon implementations use the explicit solutions that were obtained for the open-loop differential game, and therefore, the differential game is not solved repeatedly.

The paper is organized as follows. In Section~\ref{formulation}, I present two differential game models based on the HK opinion dynamics. In Section~\ref{main}, I derive the open-loop Nash equilibrium solution and the associated opinion trajectory with it for both differential game problems. The receding horizon implementations of the open-loop strategies were presented in Section~\ref{receding}. In Section~\ref{simulation}, the results from the previous sections are implemented on a real-world social network to observe the evolution of individuals' opinions when they are non-stubborn and stubborn. Conclusions and future works are discussed in Section~\ref{conclusions}.

\section{Differential Game Formulation of Opinion Dynamics}\label{formulation}

Consider a social network of $n$ agents indexed $1$ through $n$ on a communication graph $\mathcal{G}(\mathcal{V},\mathcal{E})$. The set of vertices $\mathcal{V}=\{1,\cdots,n\}$ corresponds to the set of agents. Each edge $(i,j) \in \mathcal{E}$ represents a mutual opinion flow between node $i$ and node $j$. The set of neighbors of vertex $i$ is denoted by $\mathcal{N}_i=\{j\in\mathcal{V}:(i,j) \mbox{ or }(j,i)\in \mathcal{E}, j\neq i\}$. I make the following assumption. The social graph $\mathcal{G}$ is connected. The connectivity of $\mathcal{G}$ means at least one globally reachable node (a root node of a spanning tree on the graph) and vice versa. In other words, $\mathcal{N}_i\neq \emptyset$ or $\lvert\mathcal{N}_i\rvert\neq 0$ for all $i\in\mathcal{V}$.

Let $x_i(t)\in \mathit{\mathbb{R}}$ be the opinion of agent $i$ at time $t\in[0,t_f]$ where $t_f$ is a terminal time. In the HK model, the evolution of $x_i(t)$ at each stage $k=0,1,\cdots$ is as follows
\begin{equation}\label{eq:HK}
    x_i(k+1)=\frac{1}{\lvert\mathcal{N}_i\rvert}\sum_{j\in\mathcal{N}_i}x_j(k).
\end{equation}
In this model, each agent's opinion at each stage is the average opinion of her graph neighbors. The graph neighbors here are equivalent to the bounded confidence concept in the original HK model. An agent with the confidence bound $\epsilon_i>0$ takes only those agents $j$ into account whose opinions differ from her own not bigger than her confidence bound. The set of such agents $j$ is denoted by $N_i(t)=\{j\in\mathcal{V}:\lvert x_i(t)-x_j(t)\rvert\leq \epsilon_i\}$ which in terms of a time-invariant social graph $\mathcal{G}(\mathcal{V},\mathcal{E})$, is equivalent to the set of the graph neighbors for agent $i$ as defined before, i.e., $\mathcal{N}_i$. 
Another interpretation for $\mathcal{N}_i$ could be the set of agents whom a particular agent $i$ trusts to share and fuse opinion~\cite{Babakhanbak}.

Define $u_i(t)\in \mathit{\mathbb{R}}$ as agent $i$'s influence effort or simply her control input at time $t\in[0,t_f]$. The HK model (\ref{eq:HK}) with the time-delayed control input in continuous time is given by
\begin{equation}\label{eq:dynamics0}
    \dot{x}_i(t)=\frac{1}{\lvert\mathcal{N}_i\rvert}\sum_{j\in\mathcal{N}_i}x_j(t)-x_i(t)+b_iu_i(t-\tau),
\end{equation}
where $\tau>0$ is an input delay and $b_i$ is a nonzero constant. 

The HK model applies to a non-stubborn social network whose population of agents does not hold any prejudices. The process of opinion formation in a rational, non-stubborn social network at the individual level can be subject to each agent trying to minimize her disagreement with her graph neighbors while, in the meantime, expending the least influence effort. An appropriate cost function that characterizes such behavior or preferences in a social network for agent $i\in \mathcal{V}$ to minimize is  
\begin{equation} \label{eq:quadratic-cost0-nonstub}
J_i= \frac{1}{\lvert\mathcal{N}_i\rvert}\sum_{j\in\mathcal{N}_i}\big(x_i(t_f)-x_j(t_f)\big)^2+\int_{0}^{t_f} r_iu_i^2(t-\tau)~\mathrm{dt}, 
\end{equation}
where $r_i\in \mathit{\mathbb{R}}$ is a positive scalar ($r_i>0$). The cost function in (\ref{eq:quadratic-cost0-nonstub}) has two terms. The first term averages the sum of disagreements between the final opinions of each agent and her graph neighbors. The second term is the weighted control or influence effort made during the entire opinion formation process. In this work, I refer to the optimization problem in (\ref{eq:dynamics0}) and (\ref{eq:quadratic-cost0-nonstub}) as the non-stubborn optimization of the continuous-time HK opinion dynamics model.  

Suppose that the agents of the social network are stubborn, meaning that they take their prejudices into account in the process of opinion formation. The initial opinions $x_{i0}$ are referred to as the agents' prejudices. Agents with prejudices (i.e., $x_{i0}\neq 0$) are stubborn. In the course of optimization, stubborn agents are trying to minimize their disagreement with their (graph) neighbors as well as their terminal time opinion deviation from their initial opinion. The cost function for a stubborn agent $i\in \mathcal{V}$ to minimize is
\begin{align} \label{eq:quadratic-cost0-stub}
\Tilde{J}_i&= \omega_i\big(x_i(t_f)-x_{i0}\big)^2+\frac{1-\omega_i}{\lvert\mathcal{N}_i\rvert}\sum_{j\in\mathcal{N}_i}\big(x_i(t_f)-x_j(t_f)\big)^2+\int_{0}^{t_f} r_iu_i^2(t-\tau)~\mathrm{dt}, 
\end{align}
where $0\leq\omega_i\leq 1$ is the stubbornness coefficient. A larger $\omega_i$ indicates agent $i$'s stronger stubbornness to her prejudices. Note that when $\omega_i=0$, (\ref{eq:quadratic-cost0-stub}) reduces to (\ref{eq:quadratic-cost0-nonstub}) while for $\omega_i=1$, it means at the terminal time the agent only is trying to minimize her final opinion deviation from her initial opinion. An agent with $\omega_i=1$ is totally stubborn. I refer to the optimization problem in (\ref{eq:dynamics0}) and (\ref{eq:quadratic-cost0-stub}) as the stubborn optimization of the continuous-time HK opinion dynamics model.  

In the control community, both optimization problems that emerged in this work are known as differential game problems \cite{engwerda2005lq}. 
In the context of a differential game, each agent of the network is referred to as a player. In this context, each player seeks the control $u_i(t)$ that minimizes her cost function $J_i$ or $\Tilde{J}_i$ with the given initial opinions, subject to the continuous-time HK opinion evolution equation (\ref{eq:dynamics0}). In other words, the players in the game seek to minimize their cost functions in order to find their control or influence strategies $u_i(t)$ while their opinions evolve according to the differential equation (\ref{eq:dynamics0}). The behavior of self-interested players in the game of opinion formation is best reflected via noncooperative game theory. Under the framework of noncooperative games, the players can not make binding agreements, and therefore, the solution (i.e., the Nash equilibrium) has to be self-enforcing, meaning that once it is agreed upon, nobody has the incentive to deviate from~\cite{van1983refinements}. Throughout the paper, I use the words user, agent, player, and individual interchangeably.

\section{Open-loop Nash Equilibrium}\label{main}

Nash equilibrium is the main solution concept in noncooperative game scenarios. A Nash equilibrium is a strategy combination of all players in the game with the property that no one can gain a lower cost by unilaterally deviating from it. The open-loop Nash equilibrium is defined as a set of admissible actions ($u_1^*,\cdots,u_n^* $) if for all admissible ($u_1,\cdots,u_n$) the inequalities
$J_i (u_1^*,\cdots,u_{i-1}^*,u_i^*\\,u_{i+1}^*,\cdots,u_n^* )\leq
J_i (u_1^*,\cdots,u_{i-1}^*,u_i,u_{i+1}^*,\cdots,u_n^* )$
hold for $i\in\{1,\cdots,n\}$ where $u_i\in\Gamma_i$ and $\Gamma_i$ is the admissible strategy set for player $i$. 
The noncooperative differential game and the unique Nash equilibrium associated with it are discussed in~\cite{engwerda2005lq}. 

In the following subsections, I present the open-loop Nash equilibrium solutions and the associated opinion trajectories with the equilibrium actions for both of the previously introduced optimization problems. Before that, I define the following vectors and matrices to restate the optimization problems in a compact form. These matrix definitions are borrowed from the standard literature on graph theory.

Define $A_i=[a_{ij}]$ where
\begin{equation*}
       a_{ij} =
\left\{
	\begin{array}{ll}
		1  & \text{if } (i,j)\in\mathcal{E} \text{ and } i\neq j,\\
		0  & \text{if } (i,j)\notin\mathcal{E} \text{ and } i\neq j, \\
            0 & \text{if } i=j, 
	\end{array}
\right.
\end{equation*}
to be the adjacency matrix for each agent $i\in\mathcal{V}$. The degree matrix is $D_i=\mathrm{diag}(0,\cdots,\lvert\mathcal{N}_i\rvert,\cdots,0)$ where "$\mathrm{diag}{·}$" stands for diagonal matrix. The graph Laplacian matrix for each agent $i$ is defined as
\begin{equation}\label{eq:Laplacian}
    L_i=D_i-A_i.
\end{equation}
The global adjacency, degree, and Laplacian matrices are $A=\sum_i A_i$,  $D=\sum_i D_i$, and $L=\sum_i L_i$. All the aforementioned matrices are symmetric. Additionally, let $W_i=\mathrm{diag}(0,\cdots,\omega_i,\cdots,0)$ and $W=\sum_i W_i$. Define vectors $x(t)=[x_1(t),\cdots,x_n(t)]^\top$, $B_i=[0,\cdots,b_i,\cdots,0]^\top$, and matrix $\Lambda=D^{-1}A-I$.

\subsection{Non-stubborn social network}
The non-stubborn optimization in (\ref{eq:dynamics0}) and (\ref{eq:quadratic-cost0-nonstub}) is restated in compact form as follows
\begin{align}\label{eq:nonstub-opt}
    &\quad\min_{u_i}~  J_i(u_i(t-\tau))= \frac{1}{\lvert\mathcal{N}_i\rvert}x^\top(t_f)L_ix(t_f)+\int_{0}^{t_f} r_iu_i^2(t-\tau)~\mathrm{dt},\quad i=1,\ldots,n, \\
    &{s.t.}\notag\\
    &\quad\dot{x}(t)=\Lambda x(t)+\sum_{i=1}^n B_iu_i(t-\tau), \quad x(0)=x_0, \nonumber
\end{align}

The presence of the Laplacian $L_i$ in the cost function above is due to its sum-of-squares property (see~\cite{Jond}).
\begin{equation*}
    x^\top(t_f) L_ix(t_f)=\sum_{j\in\mathcal{N}_i}\big(x_i(t_f)-x_j(t_f)\big)^2.
\end{equation*}

\begin{theorem}\label{theorem:nonstub}
 The unique Nash equilibrium actions and the associated opinion trajectory with these equilibrium actions for opinion formation in a non-stubborn social network as the noncooperative differential game in (\ref{eq:nonstub-opt}) are given by 
 \begin{align}
         &u_i(t)=-\frac{1}{r_i\lvert\mathcal{N}_i\rvert}\hat{B}_i^\top\mathrm{e}^{(t_f-t)\Lambda^\top}\hat{L}_iH^{-1}(t_f-\tau)\mathrm{e}^{(t_f-2\tau)\Lambda}x(\tau), \label{eq:Nash-nonstub}\\
         &x(t+\tau)=\Big(\mathrm{e}^{t\Lambda}-\mathrm{e}^{\tau\Lambda}\Psi(t)\Delta H^{-1}(t_f-\tau)\mathrm{e}^{(t_f-2\tau)\Lambda}\Big)x(\tau),  \label{eq:Nash-tr-nonstub} 
 \end{align}
 where
 \begin{align}
     &H(t_f-\tau)=I+\Psi(t_f-\tau)\Delta, \label{eq:matrix-H}\\ 
     &\Psi(t)=[\Psi_1(t),\cdots,\Psi_n(t)],\label{eq:matrix-Psi}\\
     &\Delta=[ \frac{1}{\lvert\mathcal{N}_1\rvert}\hat{L}_1,\cdots, \frac{1}{\lvert\mathcal{N}_n\rvert}\hat{L}_n]^\top,\label{eq:matrix-Delta} \\
     & \Psi_i(t,0)=\int_0^{t}\mathrm{e}^{(t-s)\Lambda}S_i\mathrm{e}^{(t-s)\Lambda^\top}~\mathrm{ds}, \label{eq:matrix-Psi-i} \\
     &S_i= \frac{1}{r_i}\hat{B}_i\hat{B}_i^\top,\quad\hat{B}_i=\mathrm{e}^{-\tau\Lambda}B_i,\quad\hat{L}_i=\mathrm{e}^{\tau\Lambda^\top} L_i\mathrm{e}^{\tau\Lambda}\label{eq:notation}.
 \end{align}
\end{theorem}

\begin{proof}
Before proceeding with applying the necessary optimality conditions, one should address the input delay. A common approach to dealing with the input-delayed linear dynamical system in (\ref{eq:nonstub-opt}) is to apply an integral transformation to obtain a delay-free linear dynamical system~\cite{Yusheng}. Below, I convert the input-delayed differential game in (\ref{eq:nonstub-opt}) to a delay-free differential game problem.

The cost function in (\ref{eq:nonstub-opt}) can be decomposed by
\begin{align} \label{eq:quadratic-cost0-nonstub0}
J_i(u_i(t))&= \frac{1}{\lvert\mathcal{N}_i\rvert}x^\top(t_f)L_ix(t_f)+\int_{0}^{\tau} r_iu_i^2(t)~\mathrm{dt}+\int_{0}^{t_f-\tau} r_iu_i^2(t)~\mathrm{dt}, 
\end{align}
The middle integral term on the right-hand side in (\ref{eq:quadratic-cost0-nonstub0}) is a constant since the control does not take place for $t\in [0, \tau[$. Therefore, the minimization of $J_i$ is equivalent to the minimization of
\begin{equation} \label{eq:quadratic-cost0-nonstub1}
J_i(u_i(t))= \frac{1}{\lvert\mathcal{N}_i\rvert}x^\top(t_f)L_ix(t_f)+\int_{0}^{t_f-\tau} r_iu_i^2(t)~\mathrm{dt}. 
\end{equation}

The solution of the linear system in (\ref{eq:nonstub-opt}) at $t+\tau$ is given by
\begin{align} \label{eq:transform-dynamics0}
    x(t+\tau)&=\mathrm{e}^{(t+\tau)\Lambda}x_0\nonumber+\int_0^{t+\tau}\mathrm{e}^{(t+\tau-s)\Lambda}\sum_{i=1}^n B_iu_i(s-\tau)~\mathrm{ds} \nonumber \\
  &=\mathrm{e}^{\tau\Lambda}\Big(\mathrm{e}^{t\Lambda}x_0+\int_0^{t}\mathrm{e}^{(t-s)\Lambda}\sum_{i=1}^n B_iu_i(s-\tau)  ~\mathrm{ds}+\int_t^{t+\tau}\mathrm{e}^{(t-s)\Lambda}\sum_{i=1}^n B_iu_i(s-\tau)~\mathrm{ds}\Big).
\end{align}
The first and second terms inside the outer bracket pair on the right-hand side in (\ref{eq:transform-dynamics0}) denote the solution of the linear system in (\ref{eq:nonstub-opt}) at $t$. Therefore, (\ref{eq:transform-dynamics0}) can be rewritten as
\begin{align}\label{eq:transform-dynamics}
x(t+\tau)&=\mathrm{e}^{\tau\Lambda}\left(x(t)+\int_t^{t+\tau}\mathrm{e}^{(t-s)\Lambda}\sum_{i=1}^n B_iu_i(s-\tau)~\mathrm{ds}\right) \nonumber \\&=\mathrm{e}^{\tau\Lambda}\left(x(t)+\int_{t-\tau}^{t}\mathrm{e}^{(t-s-\tau)\Lambda}\sum_{i=1}^n B_iu_i(s)~\mathrm{ds}\right)\nonumber \\&=\mathrm{e}^{\tau\Lambda} y(t),
\end{align}
where $y(t)$ is a delay-free transformation. 
By letting $t=0$ and $t=t_f-\tau$, we obtain  $\mathrm{e}^{-\tau\Lambda}x(\tau)= y(0)$ and $x(t_f)= \mathrm{e}^{\tau\Lambda}y(t_f-\tau)$, respectively. By substituting (\ref{eq:transform-dynamics}) into the cost function in (\ref{eq:nonstub-opt}), I have
\begin{align} \label{eq:quadratic-cost0-nonstub-f}
J_i(u_i(t))&= \frac{1}{\lvert\mathcal{N}_i\rvert}\big(\mathrm{e}^{\tau\Lambda}y(t_f-\tau)\big)^\top L_i\big(\mathrm{e}^{\tau\Lambda}y(t_f-\tau)\big)+\int_{0}^{t_f-\tau} r_iu_i^2(t)~\mathrm{dt} \nonumber\\
&= \frac{1}{\lvert\mathcal{N}_i\rvert}y^\top(t_f-\tau)\big(\mathrm{e}^{\tau\Lambda^\top} L_i\mathrm{e}^{\tau\Lambda}\big)y(t_f-\tau)+\int_{0}^{t_f-\tau} r_iu_i^2(t)~\mathrm{dt}. 
\end{align}

The time derivative of (\ref{eq:transform-dynamics}) is given by
\begin{align*}
    \mathrm{e}^{\tau\Lambda}\dot{y}(t)=\dot{x}(t+\tau) 
    =\Lambda x(t+\tau)+\sum_{i=1}^n B_iu_i(t),
\end{align*}
or equivalently,
\begin{align}\label{eq:transform-dynamics1}
    \dot{y}(t)=\mathrm{e}^{-\tau\Lambda}\Lambda x(t+\tau)+\sum_{i=1}^n \mathrm{e}^{-\tau\Lambda}B_iu_i(t).
\end{align} 
Substituting $x(t+\tau)=\mathrm{e}^{\tau\Lambda} y(t)$ into (\ref{eq:transform-dynamics1}) finally yields the delay-free dynamics 
\begin{align}\label{eq:delayfree-dynamics}
    \dot{y}(t)&=\mathrm{e}^{-\tau\Lambda}\Lambda \mathrm{e}^{\tau\Lambda}y(t)+\sum_{i=1}^n \mathrm{e}^{-\tau\Lambda}B_iu_i(t).
\end{align}

Note that if matrices $A_0$ and $B_0$ commute, i.e., $A_0B_0=B_0A_0$, then $B_0\mathrm{e}^{A_0}=\mathrm{e}^{A_0}B_0$ and $\mathrm{e}^{A_0}\mathrm{e}^{B_0}=\mathrm{e}^{A_0+B_0}$ (see Fact 11.14.2. in~\cite{bernstein2009matrix}). Thereby, (\ref{eq:delayfree-dynamics}) is simplified as
\begin{align}\label{eq:delayfree-dynamics-f}
    \dot{y}(t)&=\Lambda y(t)+\sum_{i=1}^n \mathrm{e}^{-\tau\Lambda}B_iu_i(t).
\end{align}

Using the new state dynamics (\ref{eq:delayfree-dynamics-f}), cost function (\ref{eq:quadratic-cost0-nonstub-f}) and matrix notations in (\ref{eq:notation}), the new optimization problem becomes 
\begin{align}\label{eq:nonstub-opt-new}
    &\quad \min_{u_i}~  J_i(u_i(t))= \frac{1}{\lvert\mathcal{N}_i\rvert}y^\top(t_f-\tau)\hat{L}_iy(t_f-\tau)+\int_{0}^{t_f-\tau} r_iu_i^2(t)~\mathrm{dt},\quad i=1,\ldots,n,\\
    &{s.t.}\notag\\
    &\quad \dot{y}(t)=\Lambda y(t)+\sum_{i=1}^n \hat{B}_iu_i(t), \quad y(0)=\mathrm{e}^{-\tau\Lambda}x(\tau). \nonumber
\end{align}

Now, I apply the conditions obtained from Pontryagin’s principle to the delay-free optimization (\ref{eq:nonstub-opt-new}).

Define the Hamiltonian 
\begin{equation}\label{eq:Hamiltonian}
        \mathcal{H}_i=r_iu_i^2(t)+\lambda_i^\top(t)\Big(\Lambda y(t)+\sum_{i=1}^n \hat{B}_iu_i(t)\Big), 
\end{equation}
where $\lambda_i(t)$ is a co-state vector.
According to Pontryagin’s principle, the necessary conditions for optimality are $\frac{\partial \mathcal{H}_i}{\partial u_i}=0$ and $\dot{\lambda}_i(t)=-\frac{\partial \mathcal{H}_i}{\partial y}$. Applying the necessary conditions on the Hamiltonian yield
\begin{align}
    &u_i(t)=-\frac{1}{r_i}\hat{B}_i^\top\lambda_i(t), \label{eq:neccesary-u}\\
    &\dot{\lambda}_i(t)=-\Lambda^\top \lambda_i(t), \label{eq:neccesary-lam}
\end{align}
with the terminal condition 
\begin{equation}\label{eq:Lambda-compact}
    \lambda_i(t_f)=\frac{1}{\lvert\mathcal{N}_i\rvert}\hat{L}_iy(t_f-\tau).
\end{equation}

The solution of (\ref{eq:neccesary-lam}) is uniquely determined by  
\begin{equation}\label{eq:LTVsoln}
\lambda_i(t)=\mathrm{e}^{(t_f-t)\Lambda^\top}\lambda_i(t_f).
\end{equation} 
Substituting this solution into (\ref{eq:neccesary-u}) yields
\begin{equation}\label{eq:Nash00}
   u_i(t)=-\frac{1}{r_i}\hat{B}_i^\top\mathrm{e}^{(t_f-t)\Lambda^\top}\lambda_i(t_f). 
\end{equation}

Substituting (\ref{eq:neccesary-u}) into the opinion dynamics in (\ref{eq:nonstub-opt}) and then using (\ref{eq:Nash00}), I get
\begin{align}
    \label{eq:dynamics-lambda}
    \dot{y}(t)&=\Lambda y(t)-\sum_{i=1}^nS_i\lambda_i(t) \nonumber\\
    &=\Lambda y(t)-\sum_{i=1}^n S_i\mathrm{e}^{(t_f-t)\Lambda^\top}\lambda_i(t_f).
\end{align}
The solution of (\ref{eq:dynamics-lambda}) at $t$ is given by
\begin{equation}\label{eq:dynamics-sol0}
    y(t)=\mathrm{e}^{t\Lambda}y(0)-\sum_{i=1}^n\Psi_i(t)\lambda_i(t_f), 
\end{equation}
where $\Psi_i(t)$ is defined in (\ref{eq:matrix-Psi-i}).

Let $\lambda(t_f)=[\lambda_1^\top(t_f),\cdots,\lambda_n^\top(t_f)]^\top$. Using notation (\ref{eq:matrix-Psi}), solution (\ref{eq:dynamics-sol0}) becomes
\begin{equation}\label{eq:dynamics-sol}
    y(t)=\mathrm{e}^{t\Lambda}y(0)-\Psi(t)\lambda(t_f). 
\end{equation}

Stacking (\ref{eq:Lambda-compact}) for $i = 1, \cdots, n$ yields
\begin{equation}\label{eq:lamtf-compact}
    \lambda(t_f)=\Delta y(t_f-\tau)
\end{equation}
with $\Delta$ defined in (\ref{eq:matrix-Delta}). Substituting $\lambda(t_f)$ from (\ref{eq:lamtf-compact}) then into (\ref{eq:dynamics-sol}) at $t_f-\tau$ yields 
 \begin{align*}
y(t_f-\tau)=\mathrm{e}^{(t_f-\tau)\Lambda}y(0)-\Psi(t_f-\tau)\Delta y(t_f-\tau),
\end{align*}
which can be rewritten as
 \begin{align}\label{eq:state-equation-simplify}
\big(I+\Psi(t_f-\tau)\Delta\big)y(t_f-\tau)=\mathrm{e}^{(t_f-\tau)\Lambda}y(0).
\end{align}

Using the notation $H(t_f-\tau)$ in (\ref{eq:matrix-H}), the compact form of (\ref{eq:state-equation-simplify}) is
 \begin{align}\label{eq:state-y(tf)}
y(t_f-\tau)=H^{-1}(t_f-\tau)\mathrm{e}^{(t_f-\tau)\Lambda}y(0).
\end{align}

If the game has a unique open-loop Nash equilibrium, then (\ref{eq:state-equation-simplify}) is satisfied for any arbitrary $y(0)$ and $y(t_f-\tau)$. Equivalently, if matrix $H^{-1}(t_f-\tau)$ has an inverse for any arbitrary $y(t_f-\tau)$, the unique equilibrium actions exist and could be calculated for all $t\in[0,t_f-\tau]$. By substituting (\ref{eq:state-y(tf)}) into (\ref{eq:Lambda-compact}) and re-substituting (\ref{eq:Lambda-compact}) in (\ref{eq:neccesary-u}), I obtain (\ref{eq:Nash-nonstub}). 

Substituting $\lambda(t_f)$ from (\ref{eq:lamtf-compact}) into (\ref{eq:dynamics-sol}) and then re-submitting $y(t_f-\tau)$ from (\ref{eq:state-y(tf)}), I have 
\begin{align}\label{eq:final-tr}
y(t)&=\mathrm{e}^{t\Lambda}y(0) -\Psi(t)\lambda(t_f) \nonumber \\
&=\mathrm{e}^{t\Lambda}y(0)-\Psi(t)\Delta y(t_f-\tau) \nonumber \\
&=\mathrm{e}^{t\Lambda}y(0)-\Psi(t)\Delta H^{-1}(t_f-\tau)\mathrm{e}^{(t_f-\tau)\Lambda}y(0). 
\end{align}
Substituting $y(t)=\mathrm{e}^{-\tau\Lambda}x(t+\tau)$ and $y(0)=\mathrm{e}^{-\tau\Lambda}x(\tau)$ into (\ref{eq:final-tr}) and after simplifications, the state trajectory for opinions in (\ref{eq:Nash-tr-nonstub}) appears. This concludes the proof.
\end{proof}

\begin{remark}\label{remark:dist} For a delay-free optimization problem (\ref{eq:nonstub-opt}), i.e., $\tau=0$, matrices $\hat{L}_i$ and $\hat{B}_i$ in the equilibrium actions (\ref{eq:Nash-nonstub}) reduce to $L_i$ and $B_i$, respectively. From the definition of the Laplacian matrix $L_i$ in (\ref{eq:Laplacian}), it can be easily figured out that all nonzero elements of $L_i$ are only in the $i$th row and column. Moreover, the nonzero elements of the $i$th row and column are at the indices $j\in\mathcal{N}_i$. 
Based on this structure, it can be deduced that the matrix product $L_iH^{-1}(t_f)\mathrm{e}^{t_f\Lambda}x_0$ only requires the initial $x_0$ entries for $\forall j\in\mathcal{N}_i$ in $x_0$. Thus, I can draw the conclusion that the equilibrium actions (\ref{eq:Nash-nonstub}) are distributed in the sense that each agent only uses local information of her own and her graph neighbors without using any global information of the communication graph.
\end{remark}

\subsection{Stubborn social network}

The stubborn optimization in (\ref{eq:dynamics0}) and (\ref{eq:quadratic-cost0-stub}) in compact form is
\begin{align}\label{eq:stub-opt}
    &\quad \min_{u_i}~  \Tilde{J}_i(x_0,u_i(t-\tau))=\nonumber\\&\big(x(t_f)- x_0\big)^\top W_i \big(x(t_f)-x_0\big)+\frac{1-\omega_i}{\lvert\mathcal{N}_i\rvert}x^\top(t_f)L_ix(t_f)+\int_{0}^{t_f} r_iu_i^2(t-\tau)~\mathrm{dt},\quad i=1,\ldots,n,\\
    &{s.t.}\notag\\
     &\quad \dot{x}(t)=\Lambda x(t)+\sum_{i=1}^n B_iu_i(t-\tau), \quad x(0)=x_0. \nonumber
\end{align}
  
\begin{theorem}\label{theorem:stub}
The unique Nash equilibrium actions and the associated opinion trajectory with these equilibrium actions for opinion formation in a stubborn social network as the noncooperative differential game in (\ref{eq:stub-opt}) are given by (\ref{eq:Nash-stub}) and (\ref{eq:Nash-tr-stub}), respectively, shown at the bottom of the next page.
 \begin{align}
         u&_i(t)=-\frac{1}{r_i}\hat{B}_i^\top\mathrm{e}^{(t_f-t)\Lambda^\top}\Big[\big(\hat{W}_i+\frac{1-\omega_i}{\lvert\mathcal{N}_i\rvert}\hat{L}_i\big)\hat{H}^{-1}(t_f-\tau)\big(\mathrm{e}^{(t_f-2\tau)\Lambda}x(\tau)+\Psi(t_f-\tau)\Omega x_0\big)-\Tilde{W}_ix_0\Big], \label{eq:Nash-stub}\\
         x&(t+\tau)=\mathrm{e}^{t\Lambda}x(\tau)-\mathrm{e}^{\tau\Lambda}\Psi(t)\Big[\big(\hat{\Delta}+\Gamma\big)\hat{H}^{-1}(t_f-\tau)\Big(\mathrm{e}^{(t_f-2\tau)\Lambda}x(\tau)+\Psi(t_f-\tau)\Omega x_0\Big)-\Omega x_0\Big],  \label{eq:Nash-tr-stub}
 \end{align}

 where
 \begin{align}
     &\Tilde{H}(t_f-\tau)=I+\Psi(t_f-\tau)\big(\hat{\Delta}+\Gamma\big), \label{eq:matrix-H-h}\\
  &\Omega=[\Tilde{W}_1,\cdots,\Tilde{W}_n]^\top, \label{eq:matrix-O}\\
    &\Gamma=[\hat{W}_1,\cdots,\hat{W}_n]^\top, \label{eq:matrix-F}\\
  &\Tilde{\Delta}=[\frac{1-\omega_1}{\lvert\mathcal{N}_1\rvert}\hat{L}_1,\cdots,\frac{1-\omega_n}{\lvert\mathcal{N}_n\rvert}\hat{L}_n]^\top,\label{eq:matrix-P}\\
  &\hat{W}_i=\mathrm{e}^{\tau\Lambda^\top}W_i\mathrm{e}^{\tau\Lambda},\quad\Tilde{W}_i= 2\mathrm{e}^{\tau\Lambda^\top}W_i. \label{eq:notation2}
 \end{align}
\end{theorem}

\begin{proof}
Substituting the transformation $x(t_f)=\mathrm{e}^{\tau\Lambda}y(t_f-\tau)$ into the cost function in (\ref{eq:stub-opt}), I have 
\begin{align}\label{eq:stub-cost-new0}
    \Tilde{J}_i(x_0,u_i(t))&=\big(\mathrm{e}^{\tau\Lambda}y(t_f-\tau)- x_0\big)^\top W_i \big(\mathrm{e}^{\tau\Lambda}y(t_f-\tau)-x_0\big)\nonumber\\&+\frac{1-\omega_i}{\lvert\mathcal{N}_i\rvert}y^\top(t_f-\tau)\hat{L}_iy(t_f-\tau)+\int_{0}^{t_f-\tau} r_iu_i^2(t)~\mathrm{dt}\nonumber\\
    &=y^\top(t_f-\tau)\mathrm{e}^{\tau\Lambda^\top}W_i\mathrm{e}^{\tau\Lambda}y(t_f-\tau)- 2y^\top(t_f-\tau)\mathrm{e}^{\tau\Lambda^\top} W_i x_0+x_0^\top W_i x_0 \nonumber\\&+\frac{1-\omega_i}{\lvert\mathcal{N}_i\rvert}y^\top(t_f-\tau)\hat{L}_iy(t_f-\tau)+\int_{0}^{t_f-\tau} r_iu_i^2(t)~\mathrm{dt}
\end{align}
Here, $x_0^\top W_i x_0$ is a constant. Then, (\ref{eq:stub-cost-new0}) is equivalent to
\begin{align}\label{eq:stub-cost-new1}
    \Tilde{J}_i(x_0,u_i(t))&=y^\top(t_f-\tau)\hat{W}_iy(t_f-\tau)-y^\top(t_f-\tau)\Tilde{W}x_0 \nonumber\\&+\frac{1-\omega_i}{\lvert\mathcal{N}_i\rvert}y^\top(t_f-\tau)\hat{L}_iy(t_f-\tau)+\int_{0}^{t_f-\tau} r_iu_i^2(t)~\mathrm{dt},
\end{align}
where $\hat{W}_i$ and $\Tilde{W}_i$ are defined in (\ref{eq:notation2}).
Therefore, the delay-free optimization becomes
\begin{align}\label{eq:stub-opt-new}
    &\min_{u_i}~\Tilde{J}_i(x_0,u_i(t))= y^\top(t_f-\tau)\hat{W}_iy(t_f-\tau)-y^\top(t_f-\tau)\Tilde{W}_ix_0\nonumber\\&\quad\quad\quad\quad\quad\quad\quad\quad+\frac{1-\omega_i}{\lvert\mathcal{N}_i\rvert}y^\top(t_f-\tau)\hat{L}_iy(t_f-\tau)+\int_{0}^{t_f-\tau} r_iu_i^2(t)~\mathrm{dt},\quad i=1,\ldots,n,\\
    &{s.t.}\notag\\
    &\quad\dot{y}(t)=\Lambda y(t)+\sum_{i=1}^n \hat{B}_iu_i(t), \quad y(0)=\mathrm{e}^{-\tau\Lambda}x(\tau). \nonumber
\end{align}

Defining the Hamiltonian in (\ref{eq:Hamiltonian}) and applying the necessary conditions for optimality yield (\ref{eq:neccesary-u}) and (\ref{eq:neccesary-lam}) with the terminal condition 
\begin{equation}\label{eq:Lambda-compact-stub}
    \lambda_i(t_f)=(\hat{W}_i+\frac{1-\omega_i}{\lvert\mathcal{N}_i\rvert}\hat{L}_i)y(t_f-\tau)-\Tilde{W}_ix_0.
\end{equation}
Stacking (\ref{eq:Lambda-compact-stub}) for $i = 1, \cdots, n$ yields
\begin{equation}\label{eq:lamtf-compact-s}
    \lambda(t_f)=\big(\Tilde{\Delta}+\Gamma \big) y(t_f-\tau)-\Omega x_0
\end{equation}
with $\Omega$, $\Gamma$, and $\hat{\Delta}$ given in (\ref{eq:matrix-O}-\ref{eq:matrix-P}), respectively.

Substituting $\lambda(t_f)$ from (\ref{eq:lamtf-compact-s}) then into (\ref{eq:dynamics-sol}) at $t_f-\tau$ yields 
 \begin{align*}
y(t_f-\tau)&=\mathrm{e}^{(t_f-\tau)\Lambda}y(0)-\Psi(t_f-\tau)\big(\Tilde{\Delta}+\Gamma \big) y(t_f-\tau)+\Psi(t_f-\tau)\Omega x_0,
\end{align*}
which can be rewritten as
 \begin{align}\label{eq:state-equation-simplify-s}
\Big(I+\Psi(t_f-\tau)\big(&\Tilde{\Delta}+\Gamma \big)\Big)y(t_f-\tau)=\mathrm{e}^{(t_f-\tau)\Lambda}y(0)+\Psi(t_f-\tau)\Omega x_0.
\end{align}

Using the notation $\hat{H}(t_f-\tau)$ in (\ref{eq:matrix-H-h}), the compact form of (\ref{eq:state-equation-simplify-s}) is
 \begin{align}\label{eq:state-y(tf)-stub}
y(t_f-\tau)=\hat{H}^{-1}(t_f-\tau)\Big(\mathrm{e}^{(t_f-\tau)\Lambda}y(0)+\Psi(t_f-\tau)\Omega x_0 \Big).
\end{align}

The game has a unique open-loop Nash equilibrium if matrix $\hat{H}(t_f-\tau)$ is nonsingular. By substituting (\ref{eq:state-y(tf)-stub}) into (\ref{eq:Lambda-compact-stub}) and re-substituting (\ref{eq:Lambda-compact-stub}) into (\ref{eq:LTVsoln}) and then into (\ref{eq:neccesary-u}), I obtain (\ref{eq:Nash-stub}). Substituting $y(t_f-\tau)$ from (\ref{eq:state-y(tf)-stub}) into (\ref{eq:lamtf-compact-s}), then re-submitting (\ref{eq:lamtf-compact-s}) into (\ref{eq:dynamics-sol}), and using $y(t)=\mathrm{e}^{-\tau\Lambda}x(t+\tau)$ and $y(0)=\mathrm{e}^{-\tau\Lambda}x(\tau)$, I get (\ref{eq:Nash-tr-stub}). The proof is complete.
\end{proof}

\begin{remark} 
Equilibrium actions (\ref{eq:Nash-stub}) could be judged similarly to Remark~\ref{remark:dist}. 
\end{remark}

\begin{remark} 
If $\omega_i=0$ for all $i=1,\cdots,n$, then the matrices $W_i$, $\hat{W}_i$, $\Tilde{W}_i$ for all $i=1,\cdots,n$, $\Omega$, and $\Gamma$ will be zero matrices and therefore, the results of Theorem~\ref{theorem:stub}, i.e., (\ref{eq:Nash-stub}) and (\ref{eq:Nash-tr-stub}) will reduces to the results of Theorem~\ref{theorem:nonstub}, i.e., (\ref{eq:Nash-nonstub}) and (\ref{eq:Nash-tr-nonstub}), respectively. 
\end{remark}

\section{Receding Horizon Implementation of Game Strategies}\label{receding}

The differential games proposed in Section~\ref{formulation} are offline optimizations of opinion formation in social networks with self-interested individuals. At the beginning of the game, the individuals determine their strategies and act accordingly over the entire time interval. The open-loop Nash equilibrium strategies in~(\ref{eq:Nash-nonstub}) and~(\ref{eq:Nash-stub}) depend on information only from the beginning of the game. However, players in the game of opinion formation naturally have a tendency to base their decisions on the current information in the network, so feedback strategies rather than open-loop strategies are in demand. Feedback information differential games are problematic due to the solvability issues of the coupled partial differential equations they convert to~\cite{EVANS45010271}.

In this work, I use the receding horizon control scheme to practice feedback strategies. In this control scheme, also known as model predictive control (MPC)~\cite{kwon2005receding} at the beginning, the optimization problem is solved to determine the open-loop Nash equilibrium strategies over the fixed time horizon, and then the first input is applied. At the next time step, the time horizon of the previously solved optimization is shifted one step forward, taking into account the current information of the network, a new optimization forms, which is then solved, and the first input is applied. This process is repeated at every future time step.

Before proposing a receding horizon implementation of open-loop Nash strategies in~(\ref{eq:Nash-nonstub}) and~(\ref{eq:Nash-stub}), I define $P_i(t_f,t)=\frac{1}{\lvert\mathcal{N}_i\rvert}\mathrm{e}^{(t_f-t)\Lambda^\top}\hat{L}_iH^{-1}(t_f-\tau)$ and rewrite~(\ref{eq:Nash-nonstub}) as
 \begin{align}
         &u_i(t)=-\frac{1}{r_i}\hat{B}_i^\top P_i(t_f,t)\mathrm{e}^{(t_f-2\tau)\Lambda}x(\tau). \label{eq:eq:Nash-nonstub-15}
 \end{align}
To rewrite the opinion trajectory $x(t+\tau)$ associated with the equilibrium actions in (\ref{eq:eq:Nash-nonstub-15}) in terms of $P_i$, I revisit (\ref{eq:dynamics-sol0}), shown at the bottom of the next page,
 \begin{align}
   y(t)&=\mathrm{e}^{t\Lambda}y(0)-\sum_{i=1}^n\Psi_i(t,0)\lambda_i(t_f)=\mathrm{e}^{t\Lambda}y(0)-\sum_{i=1}^n\int_0^{t}\mathrm{e}^{(t-s)\Lambda}S_i\mathrm{e}^{(t-s)\Lambda^\top}\lambda_i(t_f)~\mathrm{ds}\nonumber\\
    &=\mathrm{e}^{t\Lambda}y(0)-\sum_{i=1}^n\int_0^{t}\mathrm{e}^{(t-s)\Lambda}S_i\mathrm{e}^{(t-s)\Lambda^\top}\frac{1}{\lvert\mathcal{N}_i\rvert}\hat{L}_iy(t_f-\tau)~\mathrm{ds}\nonumber\\
    &=\mathrm{e}^{t\Lambda}y(0)-\sum_{i=1}^n\int_0^{t}\mathrm{e}^{(t-s)\Lambda}S_i\mathrm{e}^{(t-s)\Lambda^\top}\frac{1}{\lvert\mathcal{N}_i\rvert}\hat{L}_iH^{-1}(t_f-\tau)\mathrm{e}^{(t_f-\tau)\Lambda}y(0)~\mathrm{ds}\nonumber\\
    &=\Big(\mathrm{e}^{t\Lambda}-\sum_{i=1}^n\int_0^{t}\mathrm{e}^{(t-s)\Lambda}S_iP_i(t,s)~\mathrm{ds}~\mathrm{e}^{(t_f-\tau)\Lambda}\Big)y(0), \label{eq:dynamics-sol00}
 \end{align}
where $P_i(t,s)=\frac{1}{\lvert\mathcal{N}_i\rvert}\mathrm{e}^{(t-s)\Lambda^\top}\hat{L}_iH^{-1}(t_f-\tau)$.
Substituting $y(t)=\mathrm{e}^{-\tau\Lambda}x(t+\tau)$ and $y(0)=\mathrm{e}^{-\tau\Lambda}x(\tau)$ into (\ref{eq:dynamics-sol00}), I have
\begin{align}
   x(t+\tau)= \Big(\mathrm{e}^{t\Lambda}-\mathrm{e}^{\tau\Lambda}\sum_{i=1}^n\int_0^{t}\mathrm{e}^{(t-s)\Lambda}S_iP_i(t,s)~\mathrm{ds}~\mathrm{e}^{(t_f-2\tau)\Lambda}\Big)x(\tau). \label{eq:dynamics-rc0}
\end{align}

At each time instant $t$, assuming that the delay-free vector $\bar{x}_0=x(t)$ or delayed vector $\bar{x}_0(\tau)=x(t+\tau)$ is available, the receding horizon cost function for non-stubborn optimization is as follows
\begin{align}\label{eq:nonstub-opt-receding}
    \min_{u_i}~  J_i(x(t)&,u_i(t-\tau))= \frac{1}{\lvert\mathcal{N}_i\rvert}x^\top(t+t_f)L_ix(t+t_f)+\int_{t}^{t+t_f} r_iu_i^2(t-\tau)~\mathrm{dt},\quad i=1,\ldots,n. 
\end{align}
Following~\cite{Gu4392487}, the receding horizon Nash equilibrium actions $\bar{u}_i(t)$ and their associated opinion trajectory $\bar{x}(t+\tau)$ are defined as
\begin{align}
         &\bar{u}_i(t)=-\frac{1}{r_i}\hat{B}_i^\top P_i(t_f,0)\mathrm{e}^{(t_f-2\tau)\Lambda}\bar{x}_0(\tau), \label{eq:Nash-nonstub-rh}\\
         &\bar{x}(t+\tau)= \Big(\mathrm{e}^{t\Lambda}-\mathrm{e}^{\tau\Lambda}\sum_{i=1}^n\int_0^{t}\mathrm{e}^{(t-s)\Lambda}S_iP_i(t,0)~\mathrm{ds}~\mathrm{e}^{(t_f-2\tau)\Lambda}\Big)\bar{x}_0(\tau). \label{eq:Nash-nonstub-tr-rh}
 \end{align}
 
In the following subsections, I present three implementations of the receding horizon scheme. 

\subsection{Fixed Social Graph}

First, I assume that each individual in the social network takes only the opinions of her immediate neighbors on a fixed social graph into account when determining the set of individuals whose opinions differ from her own not bigger than her confidence bound. So that the social graph stays connected, I assume that each agent chooses a confidence bound so that at least one neighbor's opinion doesn't differ by more than her confidence bound. The proposed receding horizon implementation of Nash equilibrium strategies on a fixed social graph is presented in Algorithm~\ref{algofixed}. Instead of applying only the first input at time $t$, the receding horizon implementation is applied over the interval $[t,t+\sigma[$ (where $\tau<\sigma<t_f$) to accommodate the delay. 

\begin{algorithm}
\caption{Receding horizon implementation of Nash equilibrium strategies on a fixed social graph over any interval $[t,t+\sigma[$ ($\tau<\sigma<t_f$)} \label{algofixed}
\begin{algorithmic}[1]
\For{$i=1,\cdots,n$}
   \For{$\forall j\in\mathcal{N}_i$}
       \If{$\lvert x_i(t)-x_j(t)\rvert\leq\epsilon_i$}
             \State $\mathcal{N}_i^{tmp}\Leftarrow {\text{add}~} j$
       \EndIf
   \EndFor
\State ${\text{form matrices}~}D_i,A_i,L_i {\text{~using}~}\mathcal{N}_i^{tmp}$   
\EndFor
\State ${\text{form matrix}~}\Lambda$ \label{algoline}
\For{$i=1,\cdots,n$}
    \State ${\text{form matrices}~}\hat{L}_i,\hat{B}_i,S_i,\Phi_i(t_f-\tau)$
\EndFor
 \State ${\text{form matrices}~}\Phi(t_f-\tau),\Delta,H(t_f-\tau)$
\State ${\text{for time interval}~[t,t+\tau[:~}$
\For{$i=1,\cdots,n$}
     \State $\bar{u}_i(s)\Leftarrow 0, s\in[t,t+\tau[$
\EndFor
\State $x(s)\Leftarrow \mathrm{e}^{(t+s)\Lambda}\bar{x}_0, s\in[t,t+\tau[$

\State $\bar{x}_0(\tau) \Leftarrow \mathrm{e}^{\tau\Lambda}\bar{x}_0$
\State ${\text{for time interval}~[t+\tau,t+\sigma[:~}$
\For{$i=1,\cdots,n$}
     \State $\bar{u}_i(s)\Leftarrow {\text{use the right side in}~}(\ref{eq:Nash-nonstub-rh}), s\in[t+\tau,t+\sigma[$
\EndFor
\State $\bar{x}(s+\tau)\Leftarrow {\text{use the right side in}~}(\ref{eq:Nash-nonstub-tr-rh}), s\in[t+\tau,t+\sigma[$
\end{algorithmic}
\end{algorithm}

\subsection{Complete Social Graph}
Here, I assume that the social graph is complete and that each agent $i$ determines the set of her neighbors by examining the opinion values of all individuals in the network. The proposed receding horizon implementation of Nash equilibrium strategies on a complete social graph is presented in Algorithm~\ref{algocomplete}. The complete social graph scenario could also be considered equivalent to a time-varying social graph in the sense that, at each receding horizon implementation, the social graph's topology varies to capture the time-varying interaction network.

\begin{algorithm}
\caption{Receding horizon implementation of Nash equilibrium strategies on a complete social graph over any interval $[t,t+\sigma[$ ($\tau<\sigma<t_f$)} \label{algocomplete}
\begin{algorithmic}[1]
\For{$i=1,\cdots,n$}
   \For{$j=1,\cdots,n$}
       \If{$i\neq j$}
           \If{$\lvert x_i(t)-x_j(t)\rvert\leq\epsilon_i$}
                \State $\mathcal{N}_i\Leftarrow {\text{add}~} j$
            \EndIf
       \EndIf
   \EndFor
\State ${\text{form matrices}~}D_i,A_i,L_i$   
\EndFor
\State ${\text{go to line~\ref{algoline} in Algorithm~\ref{algofixed}}}$
\end{algorithmic}
\end{algorithm}

\subsection{Second Neighborhood in Social Graph}
In graph theory, the second neighborhood refers to all the nodes that are exactly two edges away from a given node in a graph~\cite{Bouya2020}. In the context of a social graph where an edge pairs two friends, the second neighborhood concept is equivalent to friends-of-friends. The second neighborhood concept is a useful tool from graph theory that helps model the structure and connectivity of social networks beyond direct connections between individuals. In the third receding horizon implementation in Algorithm~\ref{algosecond}, each agent, besides her immediate graph neighbors, also takes the agents in her second neighborhood into account.

\begin{algorithm}
\caption{Receding horizon implementation of Nash equilibrium strategies based on the second neighborhood over any interval $[t,t+\sigma[$ ($\tau<\sigma<t_f$)} \label{algosecond}
\begin{algorithmic}[1]
\For{$i=1,\cdots,n$}
   \For{$\forall j\in\mathcal{N}_i$}
              \If{$\lvert x_i(t)-x_j(t)\rvert\leq\epsilon_i$}
                  \State $\mathcal{N}_i^{tmp}\Leftarrow {\text{add}~} j$
           \EndIf
      \For{$\forall k\in\mathcal{N}_j$}
           \If{$\lvert x_i(t)-x_k(t)\rvert\leq\epsilon_i$}
                  \State $\mathcal{N}_i^{tmp}\Leftarrow {\text{add}~} k$
           \EndIf
      \EndFor
   \EndFor
\State ${\text{form matrices}~}D_i,A_i,L_i {~\text{using}~}\mathcal{N}_i^{tmp}$ 
\EndFor
\State ${\text{go to line~\ref{algoline} in Algorithm~\ref{algofixed}}}$
\end{algorithmic}
\end{algorithm}

In all receding horizon implementations of game strategies in Algorithm~\ref{algofixed}-~\ref{algosecond}, I assume that at each execution time, each agent $i$ adopts a confidence bound value $\epsilon_i$ that guarantees a non-empty set $\mathcal{N}_i$, so the connectivity of the social graph is always preserved. Note that the receding horizon implementations above accommodate social networks with a heterogeneous population of users with different confidence bounds. I should also mention that the receding horizon implementation for a stubborn network is a more challenging task that is left for future work.

\section{Simulations}\label{simulation}

In this section, I apply the theoretical results to the well-known social network of Zachary's Karate Club~\cite{Zachary}. With its underlying social graph given in Fig.~\ref{fig:karate}, Zachary's Karate Club network, created by Wayne Zachary in 1977, is a social network between the members of a karate club at a US university. This network has 34 nodes and 78 edges, where the nodes and edges represent the club members and their mutual friendships. 

\begin{figure}
\centering
       \includegraphics[width=0.5\textwidth]{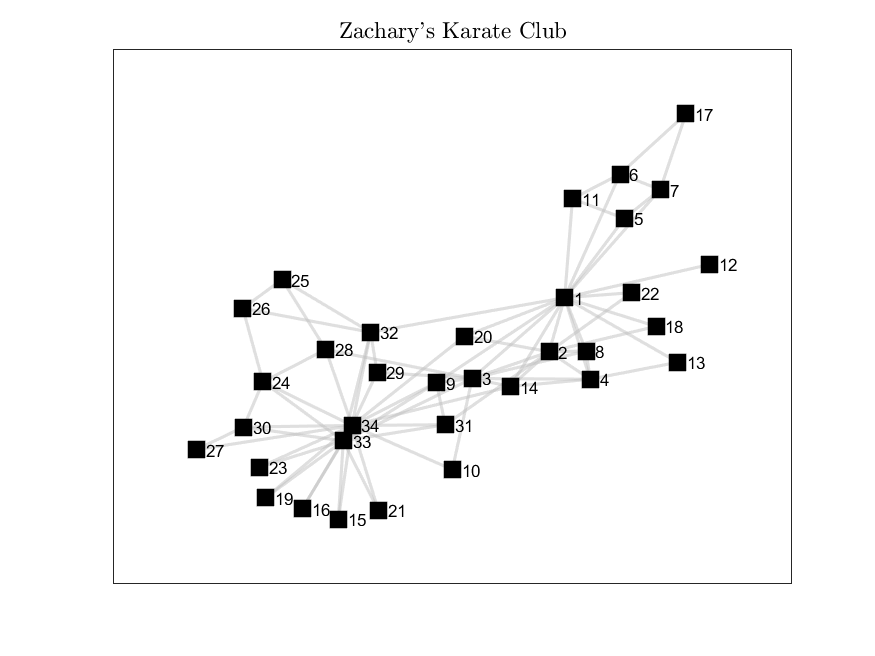}
\caption{Zachary's Karate Club network. }
\label{fig:karate}
\end{figure} 

The continuous-time HK model in (\ref{eq:dynamics0}) or its compact equation in (\ref{eq:nonstub-opt}) can be used to illustrate the formation of a particular opinion over time. The time interval is chosen $t_f=10$ and the users' opinions are initially uniformly spaced on the interval $[-1,1)$, i.e., $x_0=\{-1,-0.94,\cdots,0.92,0.98\}$. The evolution of club members' opinions prior to any optimization (i.e., $u_i(t)=0$) under the receding horizon scheme in Algorithm~\ref{algofixed} is shown in Fig.~\ref{fig:HK}. As it is seen, the club members' opinions reach a consensus about the average opinion in the network. In the first sub-figure from the left, all individuals have adopted the confidence bound value $\epsilon_i=1.2$. Running Algorithm~\ref{algofixed} multiple times showed that for a lower value than $\epsilon_i=1.2$ on a fixed social graph using the initial opinion vector above, the set of neighbors $\mathcal{N}_i$ for some individuals becomes empty and their opinions cannot be calculated. The next two sub-figures depict the opinion trajectories for the situation when the confidence bound value is increased to $\epsilon_i=1.5$ and $\epsilon_i=2$, respectively, where the effect on opinions is insignificant. 

\begin{figure*}
\centering
\begin{minipage}[b]{0.9\textwidth}
       \includegraphics[width=0.32\textwidth]{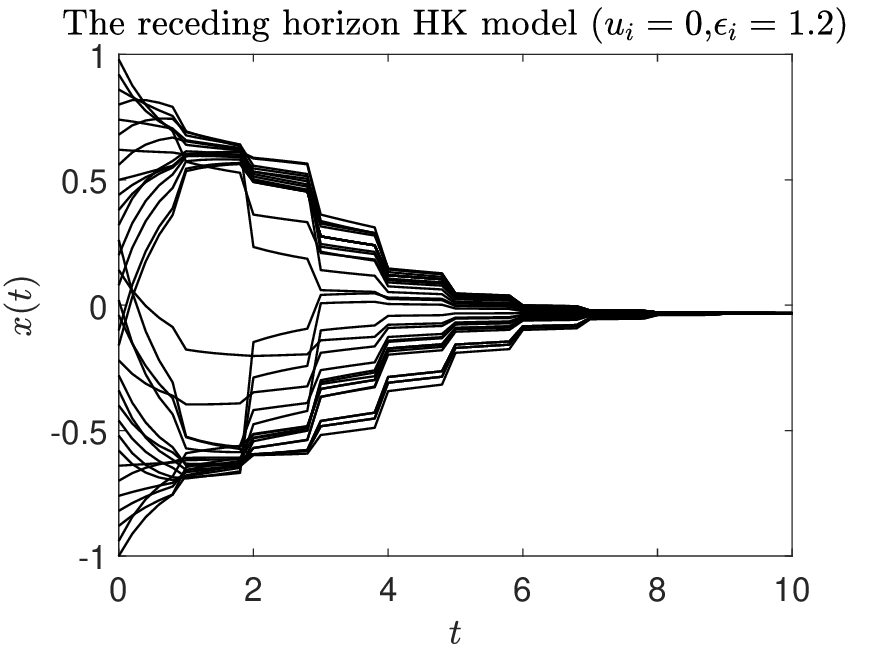}
       \includegraphics[width=0.32\textwidth]{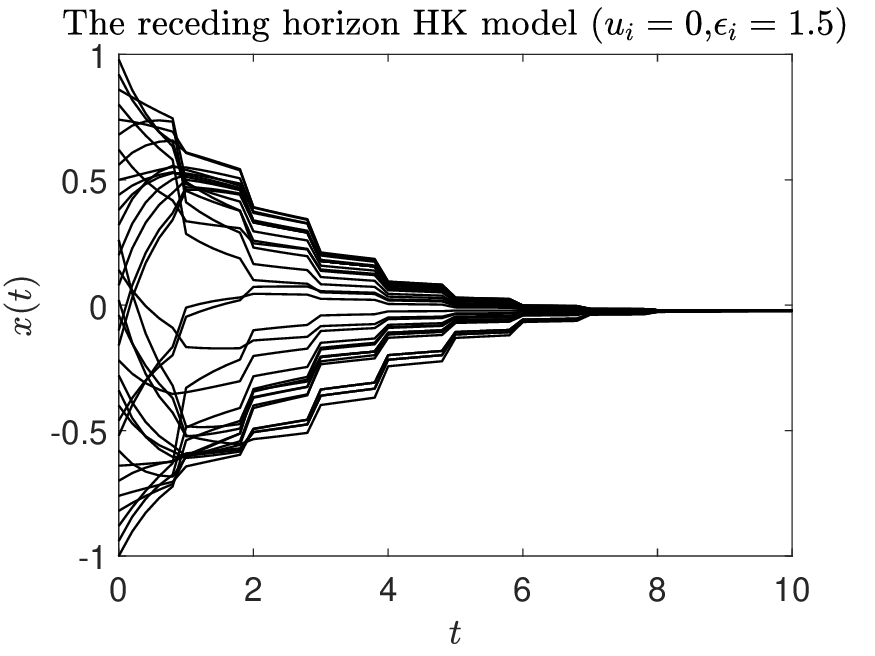}
       \includegraphics[width=0.32\textwidth]{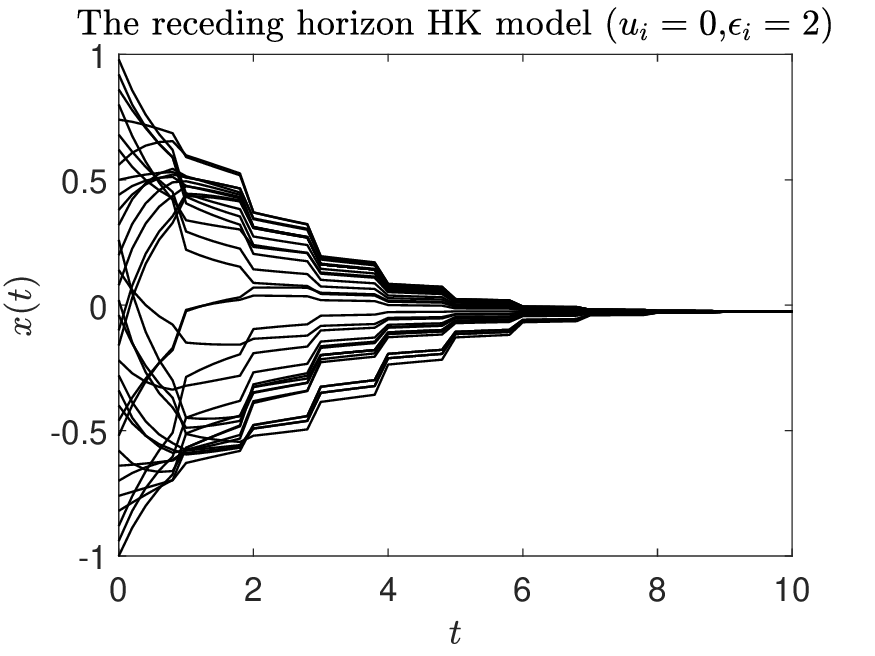}
\end{minipage}
\caption{Opinion trajectories associated with the HK model for Zachary's network prior to optimization.}
\label{fig:HK}
\end{figure*} 

For self-interested members of Zachary's Karate Club social network, whose underlying opinion dynamics is the HK model, the open-loop game strategies in Theorem~\ref{theorem:nonstub} and Theorem~\ref{theorem:stub} can be used to illustrate their opinion formation. Let $b_i=1$. For the stubborn and totally stubborn ($\omega_i=1~\forall i\in\mathcal{V}$) Zachary's network, the opinion trajectories as a result of optimization (\ref{eq:stub-opt}) are shown in Fig.~\ref{fig:stub}. As captured in (\ref{eq:quadratic-cost0-stub}), the prejudices of the stubborn members influence the formation of their final opinions. It can be observed that there is disagreement among the final opinions in the (totally) stubborn network. A large enough $r_i~\forall i\in\mathcal{V}$ can force a consensus, as the individuals in the stubborn network will try to minimize their influence effort rather than their disagreement with others and their prejudices. The rationale behind this is that with a large enough $r_i$, the corresponding optimization boils down to the minimization of the influence effort term.

\begin{figure*}
\centering
\begin{minipage}[b]{0.9\textwidth}
       \includegraphics[width=0.32\textwidth]{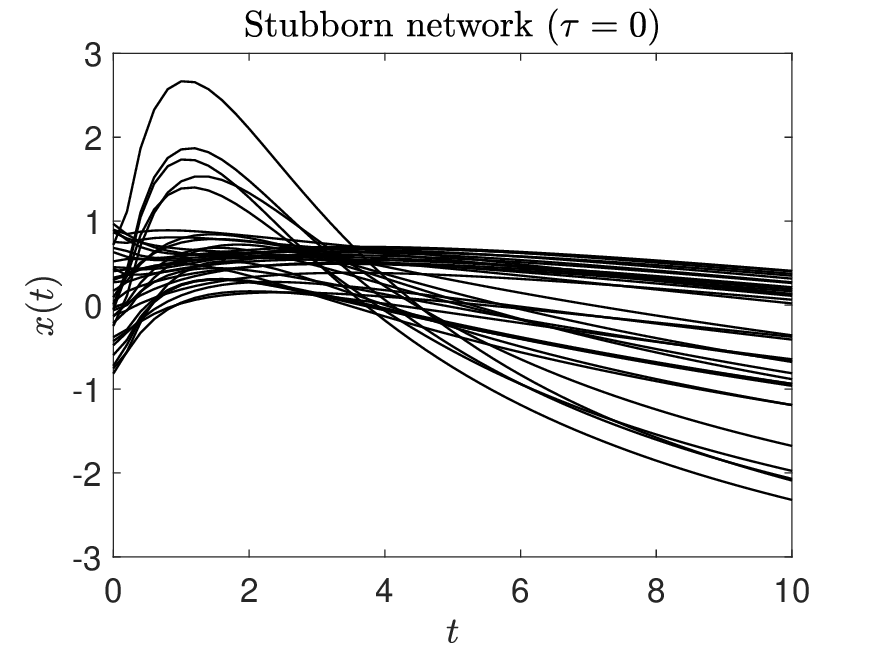}
       \includegraphics[width=0.32\textwidth]{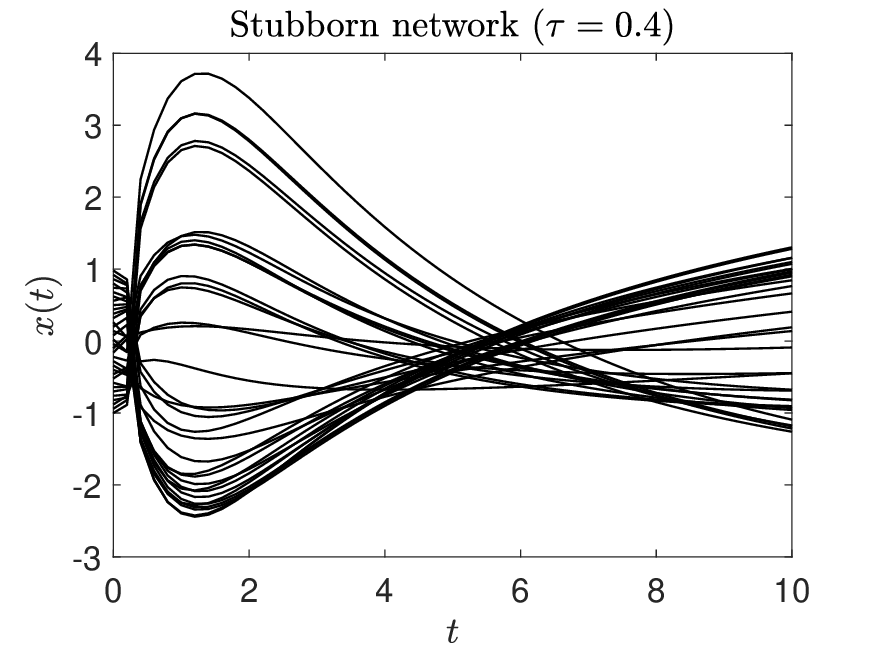}
       \includegraphics[width=0.32\textwidth]{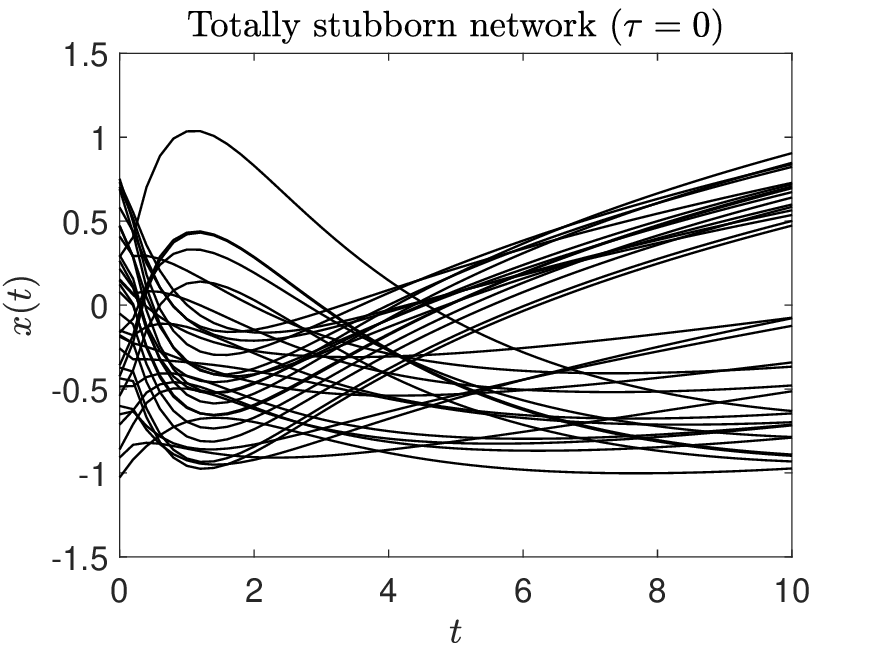}
\end{minipage}
\caption{Opinion trajectories associated with the game strategies for the stubborn and totally stubborn Zachary's network.}
\label{fig:stub}
\end{figure*}

Consider the non-stubborn Zachary's network with a fixed graph, illustrated in Fig.~\ref{fig:karate}. The resulting opinion trajectories from the receding horizon implementation in Algorithm~\ref{algofixed} are shown in Fig.~\ref{fig:fixedgraph}. Interestingly, for a delay-free network, the final opinions form two clusters, while introducing a small delay $\tau=0.4$ reduces the clusters to one, and for a slightly bigger delay $\tau=0.6$, a consensus is reached. 

\begin{figure*}
\centering
\begin{minipage}[b]{0.9\textwidth}
       \includegraphics[width=0.32\textwidth]{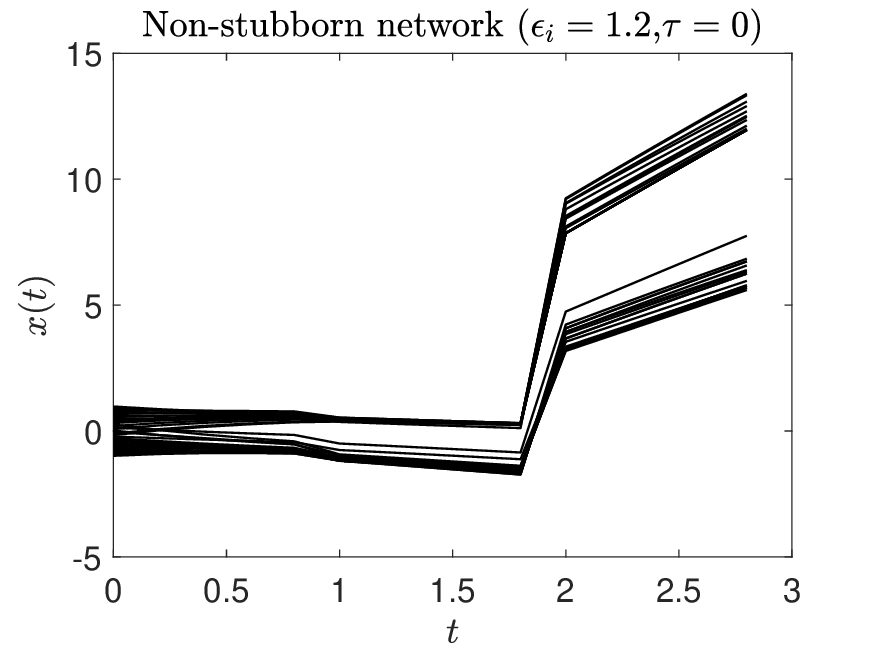}
       \includegraphics[width=0.32\textwidth]{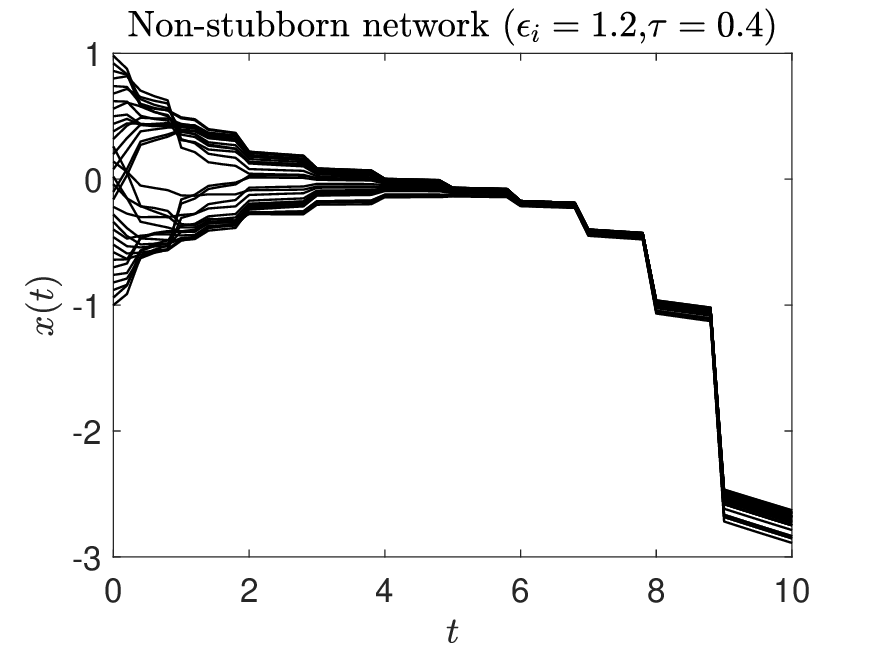}
       \includegraphics[width=0.32\textwidth]{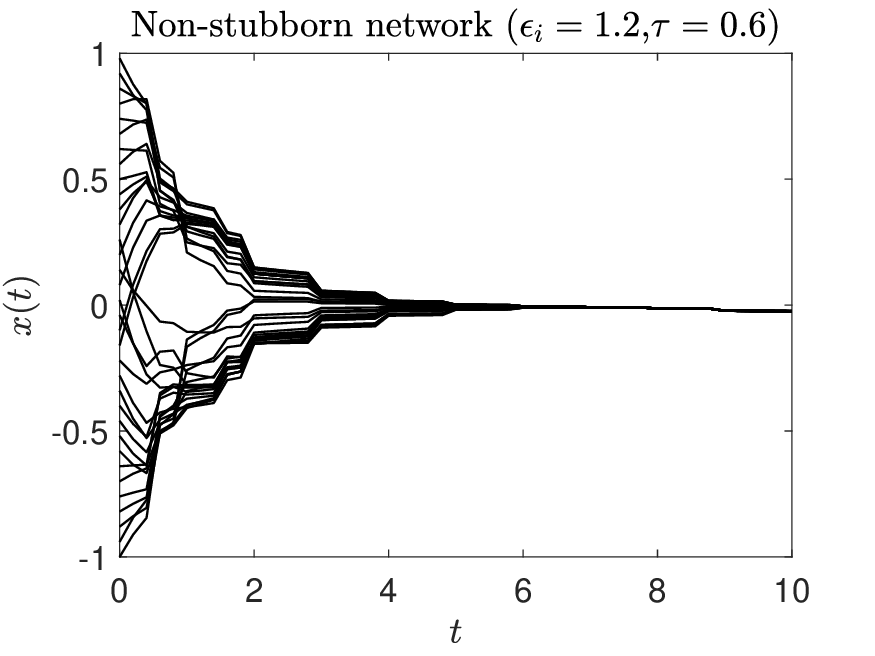}
\end{minipage}
\caption{Opinion trajectories associated with the game strategies for non-stubborn Zachary's network on a fixed graph.}
\label{fig:fixedgraph}
\end{figure*} 

For the non-stubborn Zachary's network on a complete graph, the resulting opinion trajectories from the receding horizon implementation in Algorithm~\ref{algocomplete} are shown in Fig.~\ref{fig:completegraph}. Although in a delay-free network, short-term clusters of opinions appear, the final opinions reach a consensus. In the presence of a delay, the consensus is realized smoothly.  

\begin{figure*}
\centering
\begin{minipage}[b]{0.9\textwidth}
       \includegraphics[width=0.32\textwidth]{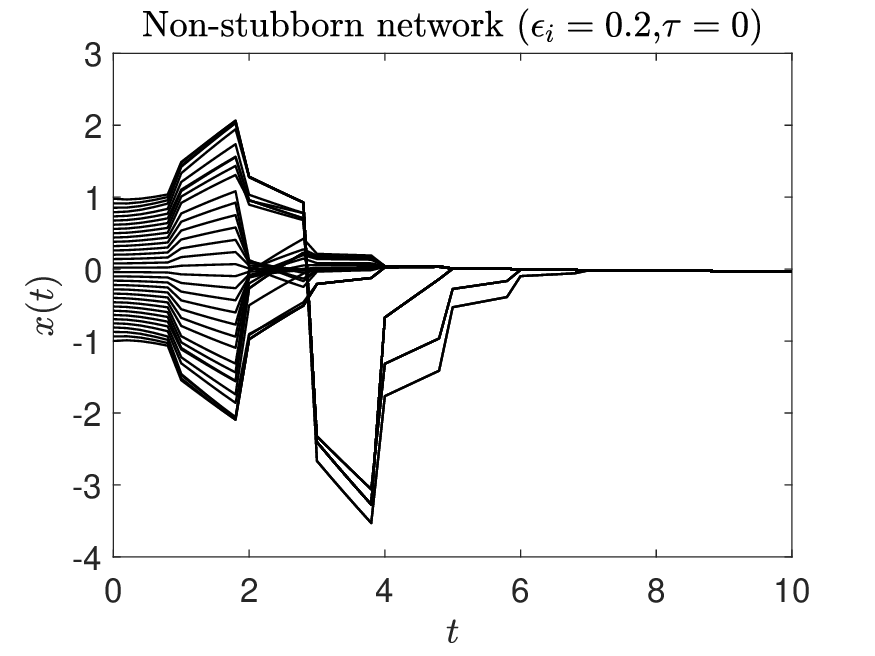}
       \includegraphics[width=0.32\textwidth]{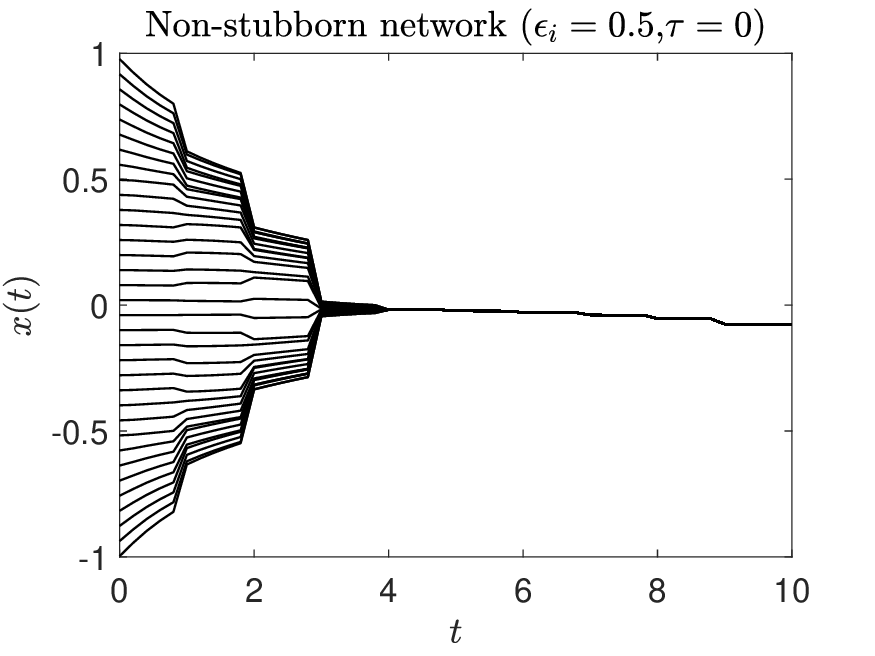}
       \includegraphics[width=0.32\textwidth]{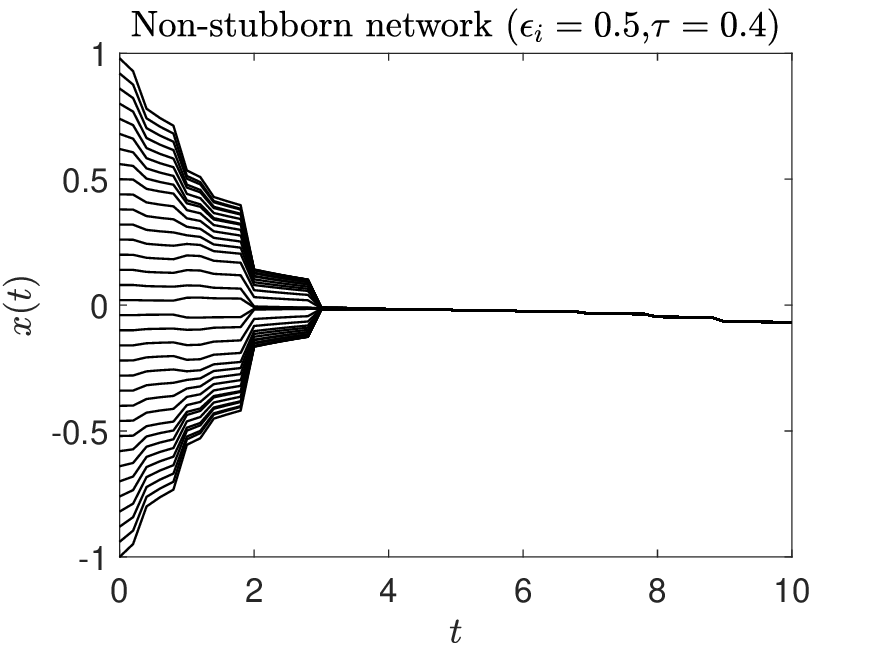}
\end{minipage}
\caption{Opinion trajectories associated with the game strategies for non-stubborn Zachary's network on a complete graph.}
\label{fig:completegraph}
\end{figure*} 

Finally, for the non-stubborn Zachary's network with a fixed graph and using the second neighborhood concept in graphs, the resulting opinion trajectories from the receding horizon implementation in Algorithm~\ref{algosecond} are shown in Fig.~\ref{fig:secondgraph}. I should note that multiple visits to the graph nodes are allowed at each execution of Algorithm~\ref{algosecond}. For a delay-free network, the final opinions polarize, while in a delayed network they reach a consensus. In the delay-free non-stubborn networks in Fig.\ref{fig:fixedgraph}-\ref{fig:secondgraph}, the confidence bound values are the lowest for which the network stays connected and the corresponding trajectories could be calculated. These values were experienced with multiple runs of their corresponding algorithms.  
\begin{figure*}
\centering
\begin{minipage}[b]{0.9\textwidth}
       \includegraphics[width=0.32\textwidth]{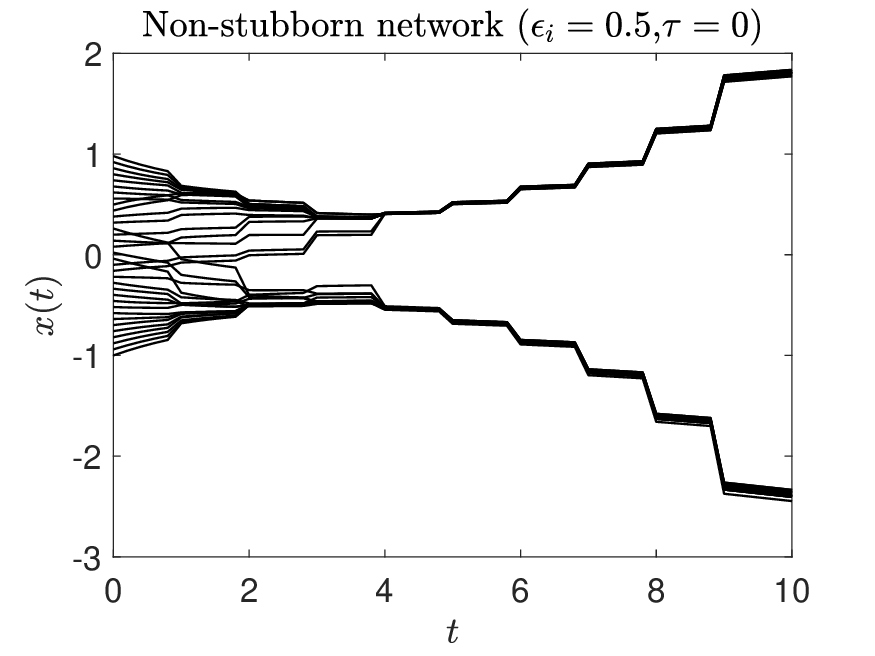}
       \includegraphics[width=0.32\textwidth]{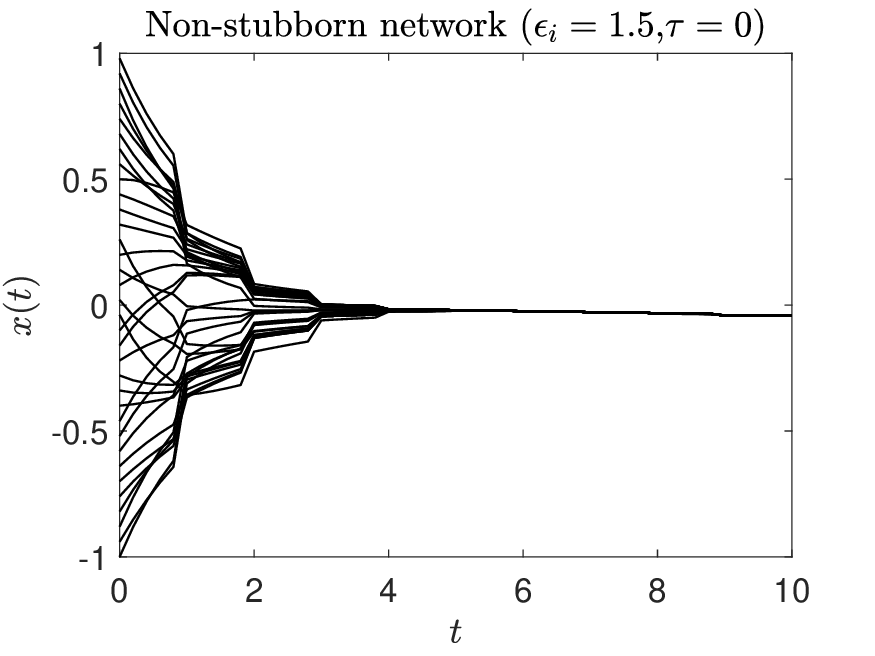}
       \includegraphics[width=0.32\textwidth]{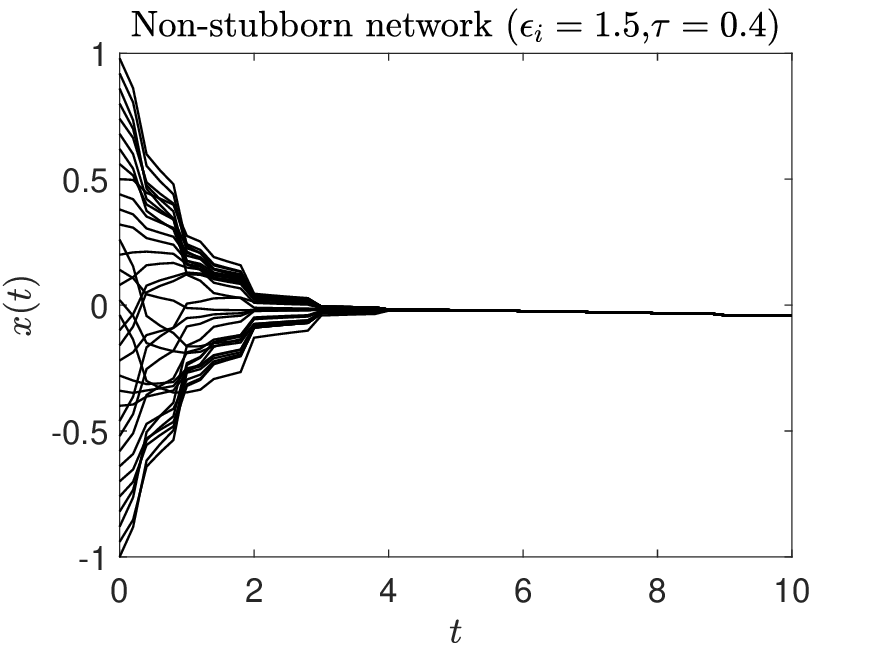}
\end{minipage}
\caption{Opinion trajectories associated with the game strategies for non-stubborn Zachary's network using the second neighborhood concept.}
\label{fig:secondgraph}
\end{figure*} 

\section{Conclusions}\label{conclusions}
In this paper, differential games were utilized to formulate opinion formation in non-stubborn and stubborn social networks, where the game's dynamics is the continuous-time Hegselmann-Krause model with a time delay in input. The explicit expressions for open-loop Nash equilibrium strategies and their associated opinion trajectories for both the non-stubborn and stubborn populations were found. Furthermore, feedback strategies were practiced for non-stubborn social networks through receding horizon implementations of open-loop game strategies. The receding horizon implementations of game strategies, depending on how big the confidence bounds and time delay are, show consensus, polarization, or clustering of opinions in the network. The receding horizon scheme for a stubborn network can be addressed in new work. A future direction of research could be analyzing the opinion dynamics in social networks on directed or signed graphs.

\section*{Acknowledgment}
This work was supported by the Czech Science Foundation (GACR) grant no. 23-07517S.

\end{document}